\nonstopmode
\documentclass[journal]{IEEEtran}
\usepackage{amsmath, amssymb, amsthm, amscd, MnSymbol}
\usepackage{mathtools}
\usepackage{epsfig}
\usepackage{subfigure}

\usepackage{cite}

\def\w{{\bf w}}
\def\s{{\bf s}}

\def\h{{\bf h}}
\def\x{{\bf x}}

\def\x{{\mathbf x}}

\newtheorem{theorem}{Theorem}[section]
\newtheorem{lemma}[theorem]{Lemma}

\begin{document}
\title{Optimal Non-coherent Data Detection for\\ Massive {SIMO} Wireless Systems:\\ A Polynomial Complexity Solution}

\author{Haider Ali Jasim Alshamary, Md Fahim Anjum, Tareq Al-Naffouri, Alam~Zaib,~Weiyu~Xu
\thanks{Haider Ali Jasim Alshamary, Md Fahim Anjum and Weiyu Xu are with the Department of Electrical and Computer Engineering, University of Iowa, IA, USA. (emails: \{haider-alshamary, mdfahim-anjum, weiyu-xu\}@uiowa.edu). Tareq Al-Naffouri is with Electrical Engineering Program, King Abdullah University of Science and Technology, Saudi Arabia. (email: tareq.alnaffouri@kaust.edu.sa). Alam Zaib is with the Electrical Engineering Department, King Fahad University of Petroleum and Minerals, Saudi Arabia. (email: alamzaib@kfupm.edu.sa).}}

\maketitle

\begin{abstract}
Massive MIMO systems can greatly increase spectral and energy efficiency over traditional MIMO systems by exploiting large antenna arrays. However, increasing the number of antennas at the base station (BS) makes the uplink noncoherent data detection very challenging in massive MIMO systems. In this paper we consider the joint maximum likelihood (ML) channel estimation and data detection problem for massive SIMO (single input multiple output) wireless systems, which is a special case of wireless systems with large antenna arrays. We propose exact ML non-coherent data detection algorithms for both constant-modulus and nonconstant-modulus constellations, with a low expected complexity. Despite the large number of unknown channel coefficients for massive SIMO systems, we show that the expected computational complexity of these algorithms is linear in the number of receive antennas and polynomial in channel coherence time. Simulation results show the performance gains (up to $5$ dB improvement) of the optimal non-coherent data detection with a low computational complexity.
\end{abstract}


\section{Introduction}

 Employing multiple-antenna arrays is well known for its benefits: high reliability, high spectral efficiency and interference reduction. Recently, a new approach, \textit{massive} MIMO, has emerged by equipping communication terminals  with a huge number of antennas. This reaps the benefits of traditional MIMO systems on a much larger scale. In \cite{Thomas}, the authors mathematically showed that the effect of fast fading and non-correlated noise is eliminated as the number of receive antennas approaches infinity. This pioneer work has generated extensive research interests in massive MIMO wireless systems. For example, massive MIMO systems' information-theoretic and propagation aspects are discussed in \cite{ScalingMIMO,MUMIMO}. Research on massive MIMO has also focused on many other aspects, including  transmit and receive schemes, the effect of pilot contamination,  energy efficiency,  and channel estimation for massive MIMO systems, as reviewed in \cite{Challenges,NextGeneration}.

 To achieve the promised advantages of massive MIMO systems, knowledge of the channel state information (CSI) is required for performing uplink data detection and downlink beamforming \cite{ScalingMIMO}. However, accurately estimating the channel coefficients is a grand challenge in wireless systems, especially in fast fading environments \cite{HowmanyAntenna} and massive MIMO system. Indeed, allocating pilot symbols to estimate time-varying channels in multi-cell massive MIMO systems will result in the issue of pilot contamination, which is a fundamental limiting factor to the performance of massive MIMO systems \cite{Thomas,NextGeneration}.

 Compared with traditional MIMO systems, it is even more challenging to perform accurate channel state estimation for massive MIMO systems, since massive MIMO systems have a large number of unknown channel coefficients. In case of conventional MIMO systems, differential modulation techniques, blind and semi-blind, and pilot based algorithms are used to solve the problem of channel tracking \cite{Stoica, HarisConference, Ma, Swindlehurst, Manton}. Although these algorithms have improved the performance of traditional non-coherent MIMO systems, they are not optimized for antenna arrays with a large number of time-varying non-coherent channels, in terms of detection performance and complexity. It is of great theoretical and practical interest to investigate near-optimal or optimal joint channel estimation and data detection schemes for massive MIMO systems \cite{NextGeneration}. For example, performing joint channel estimation and data detection will help alleviate the pilot contamination issues in multi-cell massive MIMO systems \cite{NextGeneration}.

In conventional MIMO systems, most existing efficient non-coherent signal detection algorithms are suboptimal in performance, compared with the exact ML non-coherent data detection algorithms. However, there are a few exceptions. For instance, the sphere decoder algorithm was used in \cite{Hassibi} and \cite{stojnicouter} to solve the joint ML non-coherent problem for SIMO wireless systems, but only for constant-modulus constellations (such as BPSK and QPSK). This sphere decoder reduces the computational complexity by restricting the ML detection search to a subset of the signal space. In \cite{Ma}, the authors also proposed sphere decoder algorithms to achieve the joint ML channel estimation and data detection for orthogonal space time block coded (OSTBC) wireless systems. In \cite{Hassibi} and \cite{Ma}, the sphere decoder algorithms were shown to achieve the exact ML non-coherent detection performance with a lower complexity than that of the exhaustive search. However, the sphere decoders proposed in \cite{Hassibi} and \cite{Ma} only work for constant-modulus constellations. In another line of work, \cite{Weiyu} proposed an exact joint ML channel estimation and signal detection algorithm for SIMO systems with general constellations. In \cite{OFDMBlind}, the authors proposed an exact ML channel estimation and data detection for OFDM wireless systems with general constellations. In addition, \cite{Dimitris} developed an exact ML non-coherent data detection algorithm for OSTBC systems with constant-modulus constellations, using recent results on efficient maximization of reduced-rank quadratic form to achieve polynomial complexity.

The sphere decoders in \cite{stojnicouter,Hassibi,Ma} and the ML decoder in \cite{Dimitris} work only for constant-modulus constellations. Furthermore, the optimal non-coherent data detection algorithms from \cite{Hassibi}, \cite{Ma} and \cite{Weiyu} did not look at the non-coherent data detection complexity as the number of receive antennas grows large in massive SIMO systems. The algorithm in \cite{Dimitris} gives an exact ML solution only when the matrix in quadratic form optimization has low rank, but this low-rank assumption does not hold for SIMO systems with a large number of receive antennas. Finding efficient exact ML non-coherent data detection algorithms for massive MIMO systems (including SIMO systems \cite{Andera}) with general constellations was open \cite{ScalingMIMO}.

 In this paper, we propose joint \emph{exact} ML channel estimation and data detection algorithms for massive SIMO systems, which work  with both constant-modulus and nonconstant-modulus constellations. Firstly, we propose efficient exact ML non-coherent data detection algorithms, for both constant-modulus and nonconstant-modulus constellations. Secondly, we theoretically show that the expected computational complexity is linear in the number of receive antennas and polynomial in channel coherence time, which is surprising considering a large number of unknown channel coefficients in massive SIMO systems. Thirdly, we propose a new ML tree search algorithm (TSA) which achieves the exact ML performance with near-optimal search complexity. To the best of our knowledge, these algorithms are the first set of low-complexity joint \emph{exact} ML non-coherent data detection algorithms for massive SIMO systems with general constellations. The only other work which provides efficient exact ML non-coherent data detection under general constellations is \cite{Weiyu}. However, the method in \cite{Weiyu} is for traditional SIMO systems with a small number of receive antennas, and can not guarantee polynomial expected complexity for massive SIMO systems. Moreover, our algorithm in this paper is fundamentally different from the approach in \cite{Weiyu}. Simulation results demonstrate significant performance gains of our optimal non-coherent data detection algorithms. As a consequence of this work, we demonstrate the exact performance gap between the optimal and suboptimal non-coherent data detection algorithms for massive SIMO systems, under both constant-modulus and nonconstant-modulus constellations.

 We remark that, although this paper focuses on discussing massive SIMO systems, our proposed algorithms can serve as building blocks for performing iterative joint channel estimation and data detection algorithms in general massive MIMO systems. This is beyond the scope of this current journal paper, and we will leave it as future work.

The rest of this paper is organized as follows. Section \ref{sec:problem} sets up the system model. Section \ref{sec:ML} presents our ML non-coherent data detection algorithm for constant-modulus constellations. This section also includes the derivation of the expected complexity of the proposed exact ML non-coherent data detection algorithms. Section \ref{sec:NonconstantModulus} presents the ML non-coherent data detection algorithm for nonconstant-modulus constellations, and derives its complexity. Section \ref{sec:newalgorithm} proposes a new tree search algorithm (TSA) for the exact ML non-coherent detection, and derives the complexity of the TSA. Simulation results are provided and discussed in Section \ref{sec:simulation}. Section \ref{sec:conclusion} concludes our paper and highlights our contributions.

\section{The Joint Channel Estimation and Signal Detection Problem}
\label{sec:problem}
Let $T$ denote the length of a data packet during which the channel remains constant. The channel output for a SIMO system with $N$ receive antennas is given by
\begin{equation}\label{system}
X=\h\s^*+W,
\end{equation}
where $\h \in \mathcal{C}^{N \times 1}$ is the SIMO channel vector,
$\s^* \in \mathcal{C}^{1 \times T}$ is the transmitted symbol
sequence, and $W \in \mathcal{C}^{N \times T}$ is an additive noise
matrix whose elements are assumed to be i.i.d. complex Gaussian
random variables. We also assume the entries of $\s^*$ are
i.i.d. symbols from a certain modulus constellation $\Omega$ (such as BPSK or 16-QAM).

We assume $\h$ as a deterministic unknown channel with no priori information known about it \cite{Stoica}\cite{Ma}.
Then, the joint ML channel estimation and data detection problem for SIMO systems is given by the following mixed optimization problem \vspace{-0.05in}
\begin{equation}
\min_{\h, \s^* \in \Omega^T}\| X-\h\s^*\|^2, \label{eq:mixed}
\end{equation}
where $\Omega^T$ denotes the set of $T$-dimensional signal vectors. From \cite{Hassibi}, the optimization of (\ref{eq:mixed}) over $\h$ is a least square problem while the optimization of (\ref{eq:mixed}) over $\s^*$ is an integer least square problem, since each element of $\s^*$ is chosen from a fixed constellation $\Omega$. By \cite{HarisConference}, for any given symbol vector $\s^*$, the channel vector $\h$ that minimizes (\ref{eq:mixed}) is
\begin{equation}
\hat{\h}=X\s (\s^*\s)^{-1}=X\s/\|\s\|^2, \label{eq:opth}
\end{equation}
Substituting (\ref{eq:opth}) into (\ref{eq:mixed}), we get
\begin{equation}
\|  X(\underbrace { I-\frac{1}{\|\s\|^2}\s\s^*) }_{=P_{\s}}
\|^2=\text{tr}(XP_{\s}X^*)=\text{tr}(XX^*)-\frac{1}{\|\s\|^2}\s^*X^*X\s. \label{eq:optmetric}
\end{equation}

Now, for the joint ML channel estimation and data detection,  we need to maximize $\frac{1}{\|\s\|^2}\s^*X^*X\s$ in (\ref{eq:optmetric}). This maximization depends on whether the constellation of the transmitted signal is constant or not. For massive SIMO wireless systems with a large number of unknown channel coefficients, we develop algorithms to achieve the exact ML non-coherent data detection with low expected complexity, for both constant-modulus and nonconstant-modulus constellations.

\section{Joint ML Channel Estimation and Data Detection Algorithm for Constant-modulus Constellation}
\label{sec:ML}
 In this section, we provide the joint ML channel estimation and data detection algorithm for constant-modulus constellation. In addition, we will show that the expected complexity of this proposed algorithm is polynomial in the channel coherence time.

\subsection{ML Non-coherent algorithm for constant-modulus constellation}
\label{subsec:constantalgorithm}
As pointed out in \cite{HarisConference}, if the modulation constellation is constant-modulus (such as
QPSK), the minimization of (\ref{eq:optmetric}) over $\s^*$ is
equivalent to solving the following problem:
\begin{equation}
\max_{\s^* \in \Omega^T}\s^*X^*X\s, \label{eq:max}
\end{equation}
The quadratic form in (\ref{eq:max}) for a constant modulus modulation can be changed into an equivalent minimization problem by using the maximum eigenvalue of $X^*X$. Thus, (\ref{eq:max}) can be represented as
\begin{equation}
\min_{\s \in \Omega^T}\s^*\underbrace{(\rho I-\frac{X^*X}{N}}_{=\Im})\s, \label{eq:min}
\end{equation}
where $\rho$ is a slightly larger value than the maximum eigenvalue of $\frac{X^*X}{N}$. One way of solving the integer least square optimization problem in (\ref{eq:min}) is by using exhaustive search over the entire signal space. However, the computational complexity of the exhaustive search is exponential in $T$. The sphere decoder was used in \cite{HarisConference} to efficiently solve (\ref{eq:min}) with a lower computational complexity than that of the exhaustive search. Instead of searching over all the hypotheses, sphere decoder proposes to only look at the lattice points within a radius $r$. More specifically, the sphere decoder only examines sequences $\s^*$ satisfying 
\begin{equation}
\s^*(\rho I-\frac{X^*X}{N})\s \leq r^2. \label{eq:spheresearch}
\end{equation}

From the way in which $\rho$ is determined, the matrix $\Im$ in (\ref{eq:min}) is positive semidefinite. Hence, we can use the Cholesky decomposition to factorize $\Im$ as
\begin{equation}
\Im=R^*R,
\label{eq:R}
\end{equation}
where $R$ is a $T \times T$ upper triangular matrix. Now using (\ref{eq:R}), we can rewrite (\ref{eq:min}) as
\begin{align}
\label{eq:min10}
  \min_{\s^* \in \Omega^T}\s^*(\rho I-\frac{X^*X}{N})\s &= \min_{\s^* \in \Omega^T}\s^*R^*R\s\notag\\
&=\min_{\s^* \in \Omega^T}\|R \s \|^2.
\end{align}
Since $R$ is an upper triangular matrix, $R\s$ can be expanded as
\begin{equation}\label{eq:Metric}
M_{\s^*}=\sum^T_{i=1} |\sum_{k=i}^{T} L_{i,k} \s_k|^2,
\end{equation}
where $M_{\s^*}$ is the metric of the transmitted vector $\s^*$, and $L_{i,k}$ is an entry of $R$ in the $i$-th row and $k$-th column.
For each $i$ between $1$ and $T$, we further define
\begin{equation}\label{eq:Mertici1}
 M_{\s^*_{i:T}}= |\sum_{k=i}^{T} L_{i,k} \s_k|^2+M_{\s^*_{i+1:T}},
\end{equation}
where the partial sequence $\s^*_{i:T}$ consists of elements $\s^*_i$, $\s^*_{i+1}$, ..., $\s^*_{T}$, $M_{\s^*_{i:T}}$ is the metric of the partial sequence $\s^*_{i:T}$, and $M_{\s^*_{T+1:T}}=0$ by default.

Now we represent the set of possible sequences in a tree structure as in \cite{HarisConference}. In this tree structure, we have $T$ layers, and we refer to $\s^*_{i:T}$ as a layer-$i$ node in the tree. A tree node $\s^*_{i+1:T}$ is the parent node of $\s^*_{i:T}$. Now we are ready to present the algorithm for joint ML channel estimation and data detection \cite{HarisConference}.\\

\noindent \emph{\textbf{Joint ML channel estimation data detection algorithm}} \\
Input: radius $r$, matrix $R$, constellation $\Omega$ and a $1 \times T$
index vector $I$
\begin{enumerate}
\item Set $i=T$, $r_{i}=r$, $I(i)=1$ and set $\s^*_{i}=\Omega(I(i))$.
\item (Computing the bounds) Compute the metric $M_{\s^*_{i:T}}$. If
$M_{\s^*_{i:T}}>r^2$, go to 3; else, go to 4;
\item (Backtracking) Find the smallest $i\leq j \leq T$ such
that $I(j)<|\Omega|$. If there exists such $j$, set $i=j$ and go to
5; else go to 6.

\item If $i=1$, store current $\s^*$, update $r^2=M_{\s^*_{i:T}}$ and go to 3; else set $i=i-1$, $I(i)=1$ and
$\s^*_{i}=\Omega(I(i))$, go to 2.

\item  Set $I(i)=I(i)+1$ and $s^*_{i}=\Omega(I(i))$.
Go to 2.

\item If any sequence $\s^*$ is ever found in Step 4, output the latest
stored full-length sequence as the ML solution; otherwise, double $r$
and go to 1.\\
\end{enumerate}

In our analysis of this algorithm for massive SIMO systems, we will slightly change the algorithm in the last step: if no sequence is ever found in Step 4, we will increase $r$ to $\infty$. We also remark that, for downlink beamforming, one can use the $\hat{\h}$ generated from (\ref{eq:opth}), plugging in the $\s^*_{i}$ output from joint ML algorithm.

\subsection{Choice of Radius $r$}
The choice of the radius $r$ has a big influence on the complexity of this ML algorithm. If $r^2$ is chosen bigger than the metric of every sequence $\tilde{\s}\in|\Omega|^T$, the ML algorithm may visit all the tree nodes under that radius. If $r^2$ is too small, the optimal sequence may have a metric larger than $r^2$, and the joint ML algorithm will search again under a new larger radius.

 In \cite{HarisConference, SphereComplexity}, the authors derived how to choose $r$ such that with a certain probability, the transmitted sequence has a metric no bigger than $r^2$. However, the choice of radius in \cite{HarisConference} is for a fixed number of receive antennas, and  for high signal-to-noise ratio (SNR).

In this paper, we quantify the choice of radius $r$ when the number of receive antennas is big, as in massive MIMO systems. In fact, we set $r^2$ as any constant $c$ such that
$$r^2 =c<  \frac{T D_{min}}{2},$$
where $D_{min}=\min\limits_{s_1\in \Omega ,s_2 \in \Omega, s_1\neq s_2} \|s_1-s_2\|^2$ is the minimum squared distance between two constellations points.

We remark that this choice of radius is different from that in \cite{HarisConference}.  More specifically, the new radius value does not depend on the high SNR approximation in \cite{HarisConference}, and works for massive SIMO systems.  In fact, one can choose the radius of $r$ to be a positive constant arbitrarily close to 0, for a large SIMO system. In the next section, we will show that, under this new radius, the joint ML channel estimation and data detection algorithm has expected polynomial computational complexity.

\subsection{Algorithm Computational Complexity}
\label{sec:complexity}

The computational complexity of the ML noncoherent data detection algorithm for SIMO systems is mainly determined by the number of visited nodes in each layer. By ``visited nodes'', we mean the partial sequences $\s^*_{i:T}$ for which the metric $M_{\s^*_{i:T}}$ is computed in the algorithm. The fewer the visited nodes, the lower computational complexity of the joint ML algorithm. In this section, we will show that the number of visited nodes in each layer will converge to a constant number for a sufficiently large number of receive antennas. To simplify complexity analysis, we further modify Step 6 of the ML algorithm in Section \ref{sec:problem}: ``If any sequence $\s^*$ is ever found in Step 4, output the latest stored full-length sequence as the ML solution; otherwise, let $r=\infty$ and go to 1''. We call such a modified decoder as ``modified sphere decoder''. This does not affect the algorithm's optimality. To analyze the computational complexity of our algorithm, we further assume the channel vector $\mathbf{h}$ has independent zero mean unit variance complex Gaussian components. In addition, we present our proof for constant-modulus constellations, and, \emph{in this subsection}, without loss of generality, we assume $\s$ has unit expected energy, i.e.,
\begin{equation}
|\mathbf{s}_{k}|^2=1, k=1,2,...,T. \label{Power}
\end{equation}

\begin{theorem}
Let $r^2$ be a positive constant smaller than $\frac{T D_{min}}{2}$. Then for the modified sphere decoder in the ML non-coherent data detection, the expected number of visited points at layer $i$  converges to $|\Omega|$ for $i\leq (T-1)$, as the number of receive antennas $N$ goes to infinity. The sphere decoder only visits one tree node at layer $i=T$.
\label{thm:ltt1}
\end{theorem}

\begin{proof}[Proof of Theorem \ref{thm:ltt1}]
The number of visited nodes at layer $i$ ($1\leq i \leq T-1$) in the joint ML algorithm is equal to $|\Omega|$,  if there is one and only one tree node $\widetilde{\s}^*_{(i+1):T}$ such that $M_{\widetilde{\s}^*_{(i+1):T}} \leq r^2$. In fact, we will prove that, the transmitted  $\s^*_{(i+1):T}$ will be the only sequence satisfying $M_{\widetilde{\s}^*_{(i+1):T}} \leq r^2$, with high probability as the number of receive antennas $N \rightarrow \infty$. To prove this, we first show this conclusion is true for the average case with $\Im_{E}=\rho_{E} I-\frac{E[X^*X]}{N}$, where $\rho_{E}$ is the maximum eigenvalue of $\frac{E[X^*X]}{N}$. Then we use the concentration results for $\frac{X^*X}{N}$ to prove  that, for $\Im=\rho I-\frac{E[X^*X]}{N}$,  the transmitted  $\s^*_{(i+1):T}$ will also be the only sequence satisfying $M_{s^*_{(i+1):T}} \leq r^2$, with high probability.

For the average case, we first derive ${E}[X^*X]$, and factorize $\rho_{{E}} I-\frac{{E}[X^*X]}{N}$ using the Cholesky decomposition. Using the upper triangular matrix generated from the Cholesky decomposition, we show that the transmitted $\s^*_{(i+1):T}$ will be the only sequence satisfying $M_{\s^*_{(i+1):T}} \leq r^2$ under $\Im=\rho_{{E}} I-\frac{{E}[X^*X]}{N}$.

In fact, we can write (\ref{system}) as
\begin{align}
   [\x_{1} \;  \x_{2} \;  \cdot \;  \cdot \;  \x_{T}] &= [\s^*_{1}\h \; \s^*_{2}\h \; \cdot  \; \cdot \; \s^*_{T}\h]+[\w_{1} \; \w_{2} \; \cdot \; \cdot \; \w_{T}]\notag\\
   &= [\s^*_{1}\h+\w_{1} \;\; \s^*_{2}\h+\w_{2} \;\; \cdot  \;\; \cdot \;\; \s^*_{T}\h+\w_{T}],\notag\\ \nonumber
\end{align}
where $\x_{i}$ is the $i$-th column vector of $X$. Then ${E}[X^*X]$ is equal to
\begin{equation}
E \left\{\begin{bsmallmatrix}(\s^*_{1}\h+\w_{1})^* \\ (\s^*_{2}\h+\w_{2})^* \\ \vdots \\ (\s^*_{T}\h+\w_{T})^*\end{bsmallmatrix}
\begin{bsmallmatrix} (\s^*_{1}\h+\w_{1}) & (\s^*_{2}\h+\w_{2}) & \cdots & (\s^*_{T}\h+\w_{T}) \end{bsmallmatrix} \right\}\notag\\ \nonumber.
\end{equation}
Since the entries of $\h$ are independent complex Gaussian random variables with unit variance and zero mean, ${{E}}[\h^* \h]= E[\sum^N_{i=1} h_{i}^*h_{i}]
=N$.  After some algebra,  we have
\begin{equation}
{E}[X^*X]/N=
\begin{bmatrix}
	\s_{1}\s^*_{1}+\sigma^2_{w} & \s_{1}\s^*_{2} & \cdots & \s_{1}\s^*_{T} \\
    \s_{2}\s^*_{1} & \s_{2}\s^*_{2}+\sigma^2_{w} & \cdots & \s_{2}\s^*_{T} \\
    \vdots & \vdots & \ddots & \vdots \\
    \s_{T}\s^*_{1} & \s_{T}\s^*_{2}  &  \cdots & \s_{T}\s^*_{T}+\sigma^2_{w} \\
    \end{bmatrix} \label{X*X}.
\end{equation}

We can see that (\ref{X*X}) is a Hermitian matrix with a full column rank. The maximum eigenvalue of $\frac{{E}[X^*X]}{N}$ is $\rho_{{E}}=T+\sigma^2_{w}$. Now we can write $A=\rho_{{E}} I- \frac{{E}[X^*X]}{N}$ as \begin{equation}
A=
\begin{bmatrix}
			T-\s_{1}\s^*_{1} & -\s_{1}\s^*_{2} & \cdots & -\s_{1}\s^*_{T} \\
             -\s_{2}\s^*_{1} & T-\s_{2}\s^*_{2} & \cdots & -\s_{2}\s^*_{T}\\
             \vdots & \vdots & \ddots &\vdots \\
             -\s_{T}\s^*_{1} & -\s_{T}\s^*_{2}& \cdots & T-\s_{T}\s^*_{T} \\
             \end{bmatrix} \nonumber.
\end{equation}
Using the Cholesky decomposition in \cite{NumbericalMathmatics}, we can decompose $(\rho_{{E}} I- \frac{{E}[X^*X]}{N})$ into $\grave{R}^*\grave{R}$ where $\grave{R}$ is the upper triangular matrix of Cholesky decomposition, and can be formed as
 \begin{equation}
  \grave{R}=\begin{bmatrix} L_{1,1} & L_{1,2}  & L_{1,3}  & \cdot  & \cdot & \ L_{1,T}\\
                    0&L_{2,2}  &L_{2,3}  &\cdot  &\cdot  & L_{2,T}\\
                    0 & 0  &L_{3,3}  &\cdot  &\cdot  & L_{3,T}\\
                    0&0  &0 &\cdot &\cdot & L_{T,T} \notag\end{bmatrix},
 \end{equation}
where $L_{i,i}= \sqrt{a_{i,i}-\sum^{i-1}_{k=1} L_{k,i}L^*_{k,i}}$, $L_{i,j}=\frac{1}{L_{i,i}}( a_{i,j}-(\sum^{i-1}_{k=1}L_{k,i} L^*_{k,j})^*)$ for $1\leq i<j\leq T$, and $a_{i,j}$ is an entry of $(\rho_{{E}} I- \frac{{E}[X^*X]}{N})$ with row index $i$, and column index $j$. Thus, $\grave{R}$ is given by (\ref{eq:Rmat}) (listed on the top of next page).
\begin{figure*}
	\begin{equation}
  		\grave{R}=\begin{bmatrix} \sqrt{T-1} & \frac{-(\s_{1}\s^*_{2})}{\sqrt{T-1}}  & \frac{-(\s_{1}\s^*_{3})}{\sqrt{T-1}}  & \cdots & \frac{-(\s_{1}\s^*_{T})}{\sqrt{T-1}}\\
                    0 & \sqrt{T-1 -\frac{1}{T-1}}  & \frac{1}{L_{2,2}} \left [-(\s_{2}\s^*_{3})- \frac{(\s_{2}\s^*_{3})}{T-1}\right]  & \cdots  & \frac{1}{L_{2,2}} \left [-(\s_{2}\s^*_{T})- \frac{(\s_{2}\s^*_{T})}{T-1}\right]\\
                    0 & 0  &\sqrt{T-1 -\frac{1}{T-1}-\frac{T}{(T-1)(T-2)}}  & \cdots  & \frac{1}{L_{3,3}} \left [-(\s_{3}\s^*_{T})- \frac{(\s_{3}\s^*_{T})}{T-1}-\frac{(\s_{3}\s^*_{T})T}{(T-1)(T-2)}\right]\\
                    \vdots & \vdots & \vdots & \ddots & \vdots \\
                    0      &   0    &   0    & \cdots & \sqrt{T-1 -\frac{1}{T-1}-\cdot-\frac{T}{(T-(T-2))(T-(T-1))}}\end{bmatrix}.
    \label{eq:Rmat}
	\end{equation}
\hrule
\end{figure*}

We can see that $L_{ii}=\sqrt{(T-1)-\sum_{j=1}^{i-1} \frac{T}{(T-(j-1))(T-j)}}$ for $1< i \leq T$. Now we can use $ \grave{R}$ in (\ref{eq:Rmat}) as the upper triangular matrix of Cholesky decomposition to solve the minimization equation in (\ref{eq:min10}). In fact, based on (\ref{eq:Metric}), the metric $M_{\s^*_{1:T}} (\grave{R})$ from (\ref{eq:min}) is
\begin{align}\label{eq:Mertici3}
M_{\s^*_{1:T}}=\s^*A\s &=\s^*(TI-\s\s^*)\s \notag\\
  &=T\s^*\s-\s^*\s\s^*\s\notag\\
  &=T^2-T^2 \notag\\
  &=0,
\end{align}
since $\s^* \s=T$. Because $M_{\s^*}=\sum^T_{i=1} |\sum_{k=i}^{T} L_{i,k} \s_k|^2$, from (\ref{eq:Mertici3}), we must have $| \sum_{k=i}^{T} L_{i,k}$ $ \s_k|^2=0$ for every $1\leq i \leq T$. This, in turn, implies that $M_{\s^*_{i:T}}=0$, and $\sum_{k=i}^{T} L_{i,k} \s_k=0$ for every $1\leq i \leq T$.  On the other hand, according to Lemma \ref{thm:ltt2} (the proof of which is provided in the appendix), for any other $\widetilde{\s}\neq \s$, $M_{\widetilde{\s}^*_{i:T}} \neq 0$,  where $i$ is the integer closest to $T$ such that $\s^*_i \neq \widetilde{\s}^*_{i}$.

\begin{lemma}
Let $\s^*$ be the transmitted data sequence. Let us consider using $\rho_{{E}}I-\frac{{E}[X^*X]}{N}$ for calculating the sequence metric.  For any $ \widetilde{\s}^*$ such that $\widetilde{\s}^*\neq \s^*$, $M_{\widetilde{\s}^*_{j:T}}\geq \frac{T D_{min}}{2}$ at any layer $j\leq i$, where $i$ is the largest integer such that $ \s^*_{i} \neq \widetilde{\s}^*_{i}$
\label{thm:ltt2}
\end{lemma}

When $i=T$,  the joint ML algorithm  will visit only $1$ tree node, namely $\s^*_{T}$, whose metric is equal to $0$,  because $\s^*_{T}$ is predetermined to resolve phase ambiguity; when $i<T$, at layer $i$, we also only have one sequence $\widetilde{\s}^*_{i:T}=\s^*_{i:T}$ such that $M_{\widetilde{\s}^*_{i:T}}= 0$. This will prove Theorem \ref{thm:ltt1}, under the assumption that $X^*X={E}[X^*X]$.

Now we proceed to prove that, with high probability, $X^*X/N$ is close to ${E}[X^*X]/N$, and thus the expected number of visited nodes under $\rho I-\frac{X^*X}{N}$ is very close to the case for $\rho_{{E}} I-\frac{{E}[X^*X]}{N}$. In fact, $\frac{(X^*X)_{i,j}}{N}$ can be written as the average of $N$ independent random variables under considered channel model:
\begin{align}\label{eq:factarizationofX}
\frac{(X^*X)_{i,j}}{N}&=\frac{(\s^*_{i}\h+\w_{i})^*(\s^*_{j}\h+\w_{j})}{N}\notag\\
&=\frac{\sum\limits_{k=1}^{N}(\s^*_{i}\h_k+\w_{k,i})^*(\s^*_{j}\h_{k}+\w_{k,j})}{N}\notag\\
&=\s_{i}\s^*_{j}\frac{\sum^{N}_{k=1} \h^*_{k}\h_{k}}{N}+\frac{\sum^{N}_{k=1} \w^*_{k,i} \w_{k,j}}{N}\notag\\
&+\frac{\s_{i} \sum^{N}_{k=1} \h^*_{k} \w_{k,j}}{N}+\frac{\s^*_{j} \sum^{N}_{k=1} \w^*_{k,i} \h_{k}}{N}, \\ \nonumber
\end{align}
where $\w_{i}$ is the $i$-th column of $W$. Then we can find the expectation and the variance of (\ref{eq:factarizationofX}) as follows:
\begin{align}\label{eq:EXP}
{E}[\frac{(X^*X)_{i,j}}{N}]&= \s_{i}\s^*_{j}\frac{\sum^{N}_{k=1}  {E}(\h^*_{k}\h_{k})}{N}+\frac{\sum^{N}_{k=1} {E}(\w^*_{k,i} \w_{k,j})}{N}\notag\\
&+\frac{\s_{i} \sum^{N}_{k=1} {E}(\h^*_{k} \w_{k,j})}{N}+\frac{\s^*_{j} \sum^{N}_{k=1} {E}(\w^*_{k,i} \h_{k})}{N}, \notag\\
&=\begin{cases}
    1+\sigma^2_{w},  & \text{if } i= j\\
    \s_{i}\s^*_{j},  & \text{otherwise}
\end{cases}
\end{align}

\begin{equation}\label{eq:VAR}
var(\frac{(X^*X)_{i,j}}{N})=(1+2\sigma^2_{w}+\sigma^4_{w})/N.
\end{equation}
We provide the proof of (\ref{eq:VAR}) in Appendix \ref{appendix:variance}.

The weak law of large numbers states that the sample mean of a random variable converges to its expectation in probability. Thus, for any pair $1\leq i,j \leq N$, for any constant $\xi>0$ and $\epsilon>0$, as $N\rightarrow \infty$, we have
\begin{equation}
P(|\frac{(X^*X)_{i,j}}{N}- \frac{{E}[(X^*X)_{i,j}]}{N}| \geq\varepsilon)\leq \xi.
\end{equation}

This means that, for any $\xi>0$ and $\epsilon>0$, as $N\rightarrow \infty$, we have
\begin{equation}
P(\|\frac{X^*X}{N}- \frac{{E}[X^*X]}{N}\|_F \leq \varepsilon)\geq 1-\xi,
\end{equation}
where $\|\cdot\|_{F}$ is the Frobenius norm.

Since $\rho$ is the maximum eigenvalue of $\frac{X^*X}{N}$, by the triangular inequality for the spectral norm
$$|\rho-\rho_E| <\|\frac{X^*X}{N}- \frac{{E}[X^*X]}{N}\|_2.$$

Since $$\|\frac{X^*X}{N}- \frac{{E}[X^*X]}{N}\|_2\leq \|\frac{X^*X}{N}- \frac{{E}[X^*X]}{N}\|_F,$$
we have
$$|\rho-\rho_{{E}}| <\|\frac{X^*X}{N}- \frac{{E}[X^*X]}{N}\|_F \leq \epsilon,$$
with probability at least $1-\xi$, as $N\rightarrow \infty$.

Using the triangular inequality for the spectral norm and the Frobenius norm, we have
$$\|\rho I -\frac{X^*X}{N}- (\rho_{E}I-\frac{{E}[X^*X]}{N})\|_2\leq 2\epsilon,$$
and
$$\|\rho I -\frac{X^*X}{N}- (\rho_{E}I-\frac{{E}[X^*X]}{N})\|_F\leq (\sqrt{T}+1)\epsilon,$$
with probability at least $1-\xi$, as $N\rightarrow \infty$.

Now since the Cholesky decomposition of  $(\rho I -\frac{X^*X}{N})$ is continuous at the point
$A=\rho_{{E}}I-\frac{{E}[X^*X]}{N}$,  for any $\epsilon>0$ and $\xi>0$, as $N\rightarrow \infty$,
$$\|R-\grave{R}\|_F \leq \epsilon$$
holds true with probability at least $1-\xi$. Thus as $N\rightarrow \infty$, for any full-length sequence $\widetilde{\s}^*$, with probability at least $1-\xi$,
$$|M_{\widetilde{\s}^*_{i:T}}^{\grave{R}}-M^{R}_{\widetilde{\s}^*_{i:T}}|=|\widetilde{\s}^*(R_{i:T}-\grave{R}_{i:T})\widetilde{\s}| \leq \|\widetilde{\s}\|^2 \|R-\grave{R}\|_F,$$
which is no bigger than $\|\widetilde{\s}\|^2 \epsilon$. Note here the superscripts $R$ and $R'$ in $M_{\widetilde{\s}^*_{i:T}}^{\grave{R}}-M^{R}_{\widetilde{\s}^*_{i:T}}$ describe which upper triangular matrix is used in calculating the metric.

Since we can take $\epsilon$ to be arbitrarily small, this means that, for a small enough $\epsilon$, the number of visited nodes per layer will also be equal to $|\Omega|$ under  matrix $\rho I -\frac{X^*X}{N}$, with probability at least $(1-\xi)$. For a small enough constant $\epsilon>0$ and any constant $\xi>0$, as $N\rightarrow \infty$, the expected number of visited nodes at layer $i$  is upper bounded by
$$|\Omega|+(1-\xi) |\Omega|^{T-i},$$ since the largest number of visited nodes at layer $i$ when $r=\infty$  is $|\Omega|^{T-i}$. Taking arbitrary small $\xi>0$, the expected number of visited nodes at layer $i$ will approach $|\Omega|$.
\end{proof}

In summary, we have shown that, under a fixed $\sigma_w^2$ or SNR, the sphere decoder can achieve an expected complexity of polynomial growth. In fact, as stated in Theorem \ref{thm:varyingSNR}, we can even lower the SNR requirement for each antenna, while still providing the ML non-coherent detection with polynomial expected complexity.

\begin{theorem}
 Let $r^2$ be a positive constant smaller than $TD_{min}/2 $. If $\sigma_{w}^2=o(\sqrt{N})$, then for the modified sphere decoder for the ML non-coherent data detection, the expected number of visited points at layer $i$ converges to $|\Omega|$ for $i\leq (T-1)$, as the number of receive antennas $N$ goes to infinity. The sphere decoder only visits one tree node at layer $i=T$. Here $o(\sqrt{N})$ means that $\lim_{N\rightarrow \infty}{\sigma_w^2/\sqrt{N}}=0$.
\label{thm:varyingSNR}
\end{theorem}

In fact, we can prove Theorem \ref{thm:varyingSNR} through the same arguments in proving Theorem \ref{thm:ltt1}, by noting that the variance $var(\frac{(X^*X)_{i,j}}{N})$ converges to $0$ as $N \rightarrow \infty$, if $\sigma_{w}^2=o(\sqrt{N})$. Since we fix the transmission power and the wireless channel model, $\sigma_{w}^2=o(\sqrt{N})$ means that the SNR per receive antenna is allowed to decrease, as long as  $\text{SNR} \sqrt{N} \rightarrow \infty$ as $N \rightarrow \infty$. For example, the $\text{SNR}$ can scale as $O(\log(\log(N))/\sqrt{N})$ as $N \rightarrow \infty$. This implies that we can achieve the ML non-coherent detection with low complexity, while increasing the energy efficiency of massive SIMO systems.

\section{Joint ML Channel Estimation and Data Detection Algorithm for Nonconstant-Modulus Constellations}
\label{sec:NonconstantModulus}

In Section \ref{sec:ML}, we introduced joint ML channel estimation and data detection algorithm for constant modulus constellations, and analyzed its expected complexity when $N \rightarrow \infty$. In this section, we extend our work to nonconstant-modulus constellation, and derived its complexity. This paper provides the first joint ML channel estimation and data detection algorithm for massive SIMO systems with nonconstant-modulus constellations with polynomial expected complexity.

 For nonconstant-modulus constellation, we can change the problem of maximizing (\ref{eq:optmetric}) to an equivalent minimization problem over $\s^*$
 \begin{equation}
\min_{\s^* \in \Omega^T} \frac{\s^*(\rho I-\frac{X^*X}{N})\s}{||\s||^2}, \label{eq:noncon-min}
\end{equation}
where, again, $\rho$ is slightly larger than the value of the maximum eigenvalue of $\frac{X^*X}{N}$. Now, $(\rho I-\frac{X^*X}{N})$ is a positive semidefinite matrix and can be factorized using Cholesky decomposition. Then, it can be shown that equation (\ref{eq:noncon-min}) can still be successfully transferred into another minimization problem
\begin{equation}
\min_{\s^* \in \Omega^T}\frac{\| R \s \|^2}{\|\s \|^2}, \label{eq:noncon-min2}
\end{equation}
where $R$ is the upper triangular matrix of Cholesky decomposition.

Since different sequences may have different energy, the $\|\s\|^2$ term in (\ref{eq:noncon-min2}) prevents us from solving this minimization problem through the regular sphere decoder approach.  As a result, solving (\ref{eq:noncon-min2}) by directly using the same approach as in Section \ref{sec:ML} is invalid for nonconstant modulus constellation.

In our new algorithm, we will instead lower bound $\frac{\| R \s \|^2}{\|\s \|^2}$ for partial sequences $\s_{i:T}$, taking sequence energy into consideration. To illustrate our new approach, we focus on the 16-QAM constellation $\Omega$, which comprises 16 points $a+b j$, where $a\in \{\pm 1, \pm 3\}$ and $b\in \{\pm 1, \pm 3\}$. Note that in this section, we do not assume constellation points of unit energy. The maximum energy of a constellation point in 16-QAM is thus $3^2+3^2=18$.

To lower bound $\frac{\|R \s \|^2}{\|\s \|^2}$,  we will divide the sequence $\s$ into two parts $\s_{1:i-1}$ and $\s_{i:T}$. For any partial sequence $\s^*_{i:T}$, we define a new metric, $\bar{M}_{\s^*_{i:T}}$ as,
 \begin{equation}\label{eq:Mertici1-non}
 \bar{M}_{\s^*_{i:T}}= \frac{M_{\s^*_{i:T}}}{18(i-1)+\|\s^*_{i:T}\|^2}.
\end{equation}
where $M_{\s^*_{i:T}}$ is the metric defined in (\ref{eq:Mertici1}). In fact, $\bar{M}_{\s^*_{i:T}}$ is a lower bound on  $\frac{M_{\s^*_{i:T}}}{\|\s^*_{1:i-1}\|^2+\|\s^*_{i:T}\|^2}$ or $\frac{\| R \s \|^2}{\|\s \|^2}$. We further notice that, for $i=1$, $$\bar{M}_{\s^*_{1:T}}=\frac{\| R \s \|^2}{\|\s \|^2}.$$
For other types of constellations, we can just replace $18$ in (\ref{eq:Mertici1-non}) by the maximum energy of a constellation point.

Following the setup above, we now give the Joint ML channel estimation data detection algorithm for nonconstant-modulus constellations, using the 16-QAM constellation as one example. Even though the problem is not an integer least square problem any more, we can still prove the optimality of our algorithm under the new metric.

\noindent \emph{\textbf{Joint ML channel estimation data detection algorithm for nonconstant-modulus constellations}} \\
Input: radius $r$, matrix $R$, constellation $\Omega$ and a $1 \times T$
index vector $I$
\begin{enumerate}
\item Set $i=T$, $r_{i}=r$, $I(i)=1$ and set $\s^*_{i}=\Omega(I(i))$.
\item (Computing the bounds) Compute the metric $\bar{M}_{\s^*_{i:T}}$. If
$\bar{M}_{\s^*_{i:T}}>r^2$, go to 3; else, go to 4;
\item (Backtracking) Find the smallest $i\leq j \leq T$ such
that $I(j)<|\Omega|$. If there exists such $j$, set $i=j$ and go to
5; else go to 6.

\item If $i=1$, store current $\s^*$, update $r^2=\bar{M}_{\s^*_{i:T}}$ and go to 3; else set $i=i-1$, $I(i)=1$ and
$\s^*_{i}=\Omega(I(i))$, go to 2.

\item  Set $I(i)=I(i)+1$ and $s^*_{i}=\Omega(I(i))$.
Go to 2.

\item If any sequence $\s^*$ is ever found in Step 4, output the latest
stored full-length sequence as the ML solution; otherwise, double $r$
and go to 1.\\
\end{enumerate}

\begin{theorem}
The proposed joint ML channel estimation and data detection algorithm outputs the correct joint ML sequence $\hat{s}^*$, under nonconstant-modulus constellations, by using the new metric in (\ref{eq:Mertici1-non}).
\label{thm:correct_nonconstant}
\end{theorem}

\begin{proof}
We note that the algorithm will terminate after a finite number of doubling the search radius $r$. Moreover, after the final time of doubling radius $r$, the radius will not increase anymore in the subsequence search. Let $\hat{\s}^*$ be the final sequence output by the algorithm. We must have, when the algorithm terminates, $r^2=\bar{M}_{\hat{\s}^*_{1:T}}$. Moreover, we can claim that any sequence $\s^*$ other than $\hat{\s}^*$ must have a partial sequence with metric no smaller than $\bar{M}_{\hat{\s}^*_{1:T}}$; otherwise, the algorithm will explore the full length sequence $\s^*$, and end up giving a final $r^2<\bar{M}_{\hat{\s}^*_{1:T}}$, which is a contradiction.

 Thus, for any sequence $\s^* \neq \hat{\s}^*$, there must be an $i$ such that, for the partial sequence $\s_{i:T}^*$,  $\bar{M}_{\s^*_{i:T}} \geq \bar{M}_{\hat{\s}^*_{1:T}}$. This implies $\bar{M}_{\s^*_{1:T}}$ is no smaller than $\bar{M}_{\hat{\s}^*_{1:T}}$, because $\bar{M}_{{\s}^*_{i:T}}$ is a lower bound on $\bar{M}_{\s^*_{1:T}}$. This proves that indeed $\hat{\s}^*$ has the smallest metric $\bar{M}_{\hat{\s}^*_{1:T}}$.

\end{proof}

 \subsection{Choice of Radius $r$}
For non-coherent massive SIMO systems, we need to provide an initial search radius which insures low computational complexity. For massive SIMO systems adopting 16-QAM, we derive the initial search radius as
\begin{equation}\label{eq:radius-non}
 r^2\leq\frac{2}{45}.
\end{equation}
 This radius insures that the optimal solution is inside the search radius with high probability. We provide the derivation of this radius (namely Lemma \ref{thm:radiusthm_nonconstant}) in Appendix \ref{sec:appendxiradius16QAM}. We also analyze the expected complexity for nonconstant-modulus constellations. In the end, we show that, even for nonconstant-modulus constellations, the expected complexity is also polynomial in channel coherence length and the number of antennas. This analysis will be similar to that of Section \ref{sec:complexity}, but more technically involved. In fact, we show that $r$ can be any constant number close to zero for a sufficiently large number of receive antennas irrespective of the SNR.

 \subsection{Computational Complexity of ML Algorithm for Nonconstant-Modulus Constellations}
 \label{sec:NonconstantModulusComplexity}
  Similar to the case of the algorithm for constant-modulus constellations, we will show that for massive SIMO systems with nonconstant-modulus constellations, as the number of receive antennas grows to infinity, the expected number of visited nodes in each layer will be a constant number, namely $|\Omega|$. Again, to simplify complexity analysis, we further modify Step 6 of the ML algorithm for nonconstant-modulus constellations: ``If any sequence $\s^*$ is ever found in Step 4, output the latest stored full-length sequence as the ML solution; otherwise, let $r=\infty$ and go to 1''.  We also further assume the channel vector $\mathbf{h}$ has independent zero mean unit variance complex Gaussian components, and assume that 16-QAM constellation is used.

 \begin{theorem}
Let $r^2$ be a positive constant smaller than $\frac{2}{45}$. For nonconstant-modulus constellation massive SIMO system with $N$ receive antennas, the expected number of visited points by the ML channel estimation and data detection algorithm at layer $i$ converges to $|\Omega|$ for $i\leq (T-1)$, as $N\rightarrow \infty$. The joint ML algorithm only visits one tree node at layer $i=T$.
\label{thm:ltt1-non}
\end{theorem}

 Taking the same analysis in Section \ref{sec:complexity}, we can write the maximum eigenvalue of the Hermitian matrix $\frac{{E}[X^*X]}{N}$ as $\rho_{{E}}=\sum^T_{k=1} \|\s_k\|^2+\sigma^2_{w}$. Then we can represent $A=\rho_{{E}} I- \frac{{E}[X^*X]}{N}$ as
\begin{equation}
A=
\begin{bmatrix}
			t-\s_{1}\s^*_{1} & -\s_{1}\s^*_{2}  & \cdots & -\s_{1}\s^*_{T} \\
             -\s_{2}\s^*_{1} & t-\s_{2}\s^*_{2} & \cdots & -\s_{2}\s^*_{T}\\
              \vdots         & \vdots           & \vdots & \vdots \\
             -\s_{T}\s^*_{1} & -\s_{T}\s^*_{2}& \cdots & t-\s_{T}\s^*_{T} \\
             \end{bmatrix} \nonumber.
\end{equation}
Where $t=\sum^T_{k=1} \|\s_k\|^2$. After decomposing $A$ using Cholesky decomposition, we can find the entries of $\grave{R}$ such that $\grave{R}^*\grave{R}$. Then, we can find an expression to the diagonal entries of the $\grave{R}$ as
\begin{equation}\label{eq:Lii-non}
L_{i,i}=\sqrt{t-||\s_{i}||^2-\sum^{i-1}_{j=1}\frac{||\s_{j}||^2||\s_{i}||^2t}{(t-||\s_{1:j-1}||^2)(t-||\s_{1:j}||^2)}}.
\end{equation}

We can find the metric  $\bar{M}_{\s^*_{1:T}} $ of the transmitted signal $\s^*$ as
\begin{align}
\bar{M}_{\s_{1:T}^*}=\frac{\s^*A\s}{\|\s\|^2} =\frac{\s^*(tI-\s\s^*)\s}{\|\s\|^2}=0, \notag
\end{align}
since $\s^* \s=t$. As a result, $\bar{M}_{\s^*_{i:T}}=0$ for any partial sequence $\s^*_{i:T}$ of the transmitted sequence $\s^*_{1:T}$. On the other hand, according to Lemma \ref{thm:radiusthm_nonconstant} (whose proof is given in the appendix), for any other signal $\widetilde{\s}\neq \s$, $\bar{M}_{\widetilde{\s}^*_{j:T}} \geq \frac{2}{45} $ at any layer $j\leq i$, where $i$ is the largest integer such that $ \s^*_{i} \neq \widetilde{\s}^*_{i}$..

\begin{lemma}
Let $\s^*$ be the transmitted data sequence. Let us consider using $\rho_{{E}}I-\frac{{E}[X^*X]}{N}$ for calculating the sequence metric. For any $ \widetilde{\s}^*$
such that $\widetilde{\s}^*\neq \s^*$, $\bar{M}_{\widetilde{\s}_{j:T}^*}\geq \frac{2}{45}$ at any layer $j\leq i$, where $i$ is the largest integer such that $ \s^*_{i} \neq \widetilde{\s}^*_{i}$.
\label{thm:radiusthm_nonconstant}
\end{lemma}

Thus if we set $r^2<\frac{2}{45}$, under the expected matrices, the ML non-coherent data detection algorithm will only visit $|\Omega|$ nodes in each layer.  Following similar concentration arguments for the matrix $\rho I- \frac{X^*X}{N}$ in the proof of Theorem \ref{thm:ltt1}, we can similarly prove Theorem \ref{thm:ltt1-non}.

\section{Tree search Algorithm}
\label{sec:newalgorithm}

In the sections above, we consider each partial sequence as a node in a tree structure of $T$ layers. The computational complexity of the earlier algorithms heavily depends on how the initial search radius $r$ is chosen. Although the search radius $r$ is chosen so that the true transmitted sequence is within the sphere with high probability, the radius does not guarantee the minimum number of visited nodes in the tree search.

In this section we design a best-first branch-and-bound tree search algorithm for ML non-coherent data detection that does not need an assigned initial radius $r$. We call this algorithm the Tree Search Algorithm (TSA). In contrast to the algorithm in Sections \ref{sec:problem}, TSA sets the initial search radius as zero at the beginning of the algorithm. Then the radius $r$ in TSA systematically increases until the joint ML solution is found. This algorithm guarantees to visit no more tree nodes than the algorithm in Sections \ref{sec:problem}. We will show that our previous complexity results also upper bound the complexity of TSA. Moreover, we prove that this new TSA applies to nonconstant-modulus constellations.

We first introduce several terminologies about the tree structure we are using. A partial sequence $\widetilde{\s}^*_{i:T}$, $1\leq i\leq T$, corresponds to a layer-$i$ node in the tree. A node $\widetilde{\s}^*_{i:T}=(\widetilde{\s}_i^*,\widetilde{\s}_{i+1:T}^* ) $ is called a child node of its parent node $\widetilde{\s}_{i+1:T}^*$.  The parent node of any layer-$T$ node $\widetilde{\s}_{T}^*$ is called the root node. In a tree, any tree node without a child node is called a leaf node. For example, in  (b) of Figure \ref{fig:TSA}, node 1 is the root node, and node 2 is the parent node of node 9.

\begin{figure}[!htb]
\centering
\subfigure[First search iteration] {\includegraphics[width=1.65 in,height=2.1 in ]{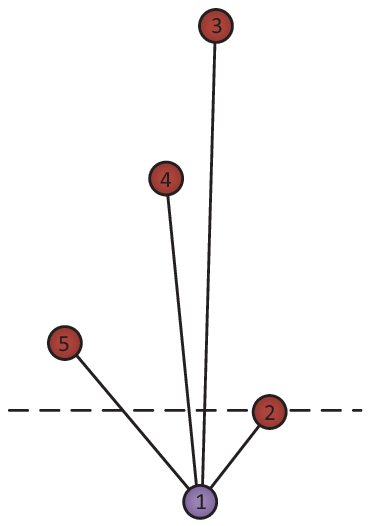}}
\subfigure[Second search iteration]{\includegraphics[width=1.65 in,height=2.1 in ]{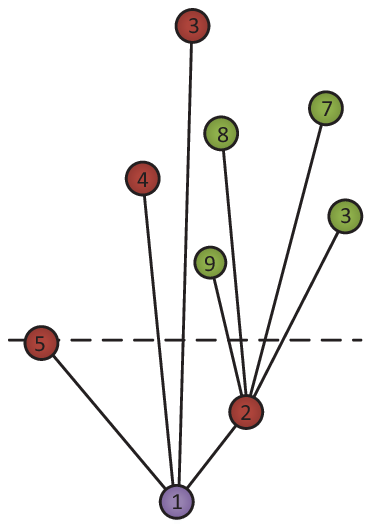}}
\subfigure[Third search iteration] {\includegraphics[width=1.65 in,height=2.1 in ]{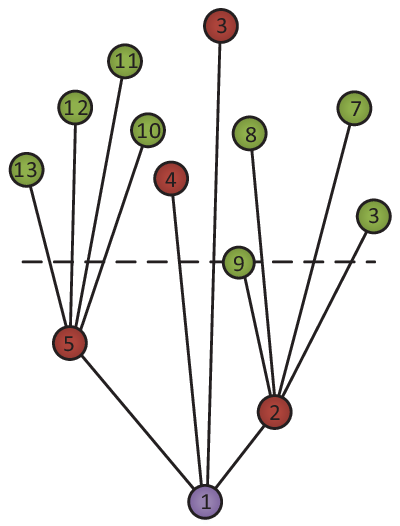}}
\caption{Illustration of tree search algorithm for a tree of $3$ layers}
\label{fig:TSA}
\end{figure}

In the TSA algorithm, we start to construct a tree which has only the root node with metric $0$. Then in each iteration,  the TSA always first finds the leaf node with the smallest metric, which is called the seed node. Then the algorithm expands the tree by adding the seed node's $|\Omega|$ child nodes to the tree, and, moreover, calculates the metrics of all these child nodes.  The tree search algorithm then iterates this process of finding the seed node and expanding the tree, until the selected seed node is a layer-$1$ node, corresponding to a full-length sequence.  The flow of this algorithm is described as below for constant-modulus constellations (for nonconstant-modulus modulations we just need to replace $M_{\widetilde{\s}^*_{i:T}}$ by $\bar{M}_{\widetilde{\s}^*_{i:T}}$ ).

\noindent \emph{\textbf{Tree search algorithm}} \\
Input: matrix $R$ and constellation $\Omega$.
\begin{enumerate}
\item  Add the root node, and set its metric to $0$. Set $r^2=0$;
\item (Find the seed node) Find the leaf node $\widetilde{\s}^*_{i:T}$ which has the smallest metric among all the leaf nodes. Select that leaf node as the seed node.
Update $r^2=M_{\widetilde{\s}^*_{i:T}}$;
\item If the seed node $\widetilde{\s}^*_{i:T}$ is layer-1 node, namely $i=1$, then go to 4; else, add the  $|\Omega|$ child nodes of $\widetilde{\s}^*_{i:T}$ to the tree, compute the metrics of these child nodes, and go to 2;
\item Terminate the algorithm, output $\widetilde{\s}^*_{1:T}$ as the optimal sequence. Output $r^2$ as the smallest possible metric.
\end{enumerate}

Figure \ref{fig:TSA} shows $3$ search iterations for QPSK constellation and $T=3$. The height of a node represents its metric. In (a), the root node $1$ is selected as the seed node, and expands into $4$ child nodes. Then node $2$ is chosen as the seed node, and expands into $4$ child nodes. The expansion of node $2$ is shown in (b). The TSA then finds node $5$ as the next seed node. The third search iteration in (c) expands node $5$ by adding its 4 children. The TSA algorithm then finds node $9$ as the seed node since it has the smallest metric. Since node $9$ is a layer-$3$ node, the algorithm will terminate and output node $9$ as the ML solution.

\subsection{Computational Complexity of TSA}
\label{sec:TSAcomplexity}
In this section, we will show that the TSA algorithm is computationally efficient in terms of the number of visited nodes.
\begin{theorem}
 The TSA outputs the optimal sequence in joint channel estimation and data detection. Let $M$ be the metric of the optimal sequence, and let $l$ be the number of sequences (including partial sequences) that have metrics no bigger than $M$. Then the number of visited points by TSA is no more than $(|\Omega|+1)l$ . Moreover, the TSA algorithm visits no more tree nodes than the sphere decoders in Section \ref{sec:ML} and \ref{sec:NonconstantModulus}.
\label{thm:TSA}
\end{theorem}

\begin{proof}
We first notice that every full-length sequence $\widetilde{\s}^*_{1:T}$ is a direct or indirect child node of a leaf node $\widetilde{\s}^*_{i:T}$ existing at the termination of the TSA. However, by the TSA, the metric $M_{\widetilde{\s}^*_{i:T}}$  must be no smaller than the final $r^2$. Since  $M_{\widetilde{\s}^*_{i:T}}$ is a lower bound of $M_{\widetilde{\s}^*_{1:T}}$,  we have  $M_{\widetilde{\s}^*_{1:T}} \geq r^2$ at the termination of the TSA. This proves that the TSA indeed outputs the optimal sequence, and $r^2=M$ at its termination.

 According to its procedure, the TSA algorithm will not visit the child nodes of any node $B$ which has a metric bigger than $M$, namely node $B$ will not be selected a seed node in the tree search. In fact, the TSA will add the full-length optimal sequence and all its (direct or indirect) parent nodes to the tree (because a parent node's metric is  always no bigger than its child node's) even before node $B$ is selected as the seed node. The TSA will then declare the full-length optimal sequence as the solution, and terminates before node $B$ is ever selected as a seed node.  So the TSA algorithm can only visit tree nodes which have metric no bigger than $M$, and possibly their direct child nodes. This gives an upper bound of $(|\Omega|+1)l$ on the total number of visited tree nodes.

To find the optimal sequence, the sphere decoder must have used a radius $r$ such that $r^2 \geq M$. Thus the sphere decoder will visit every tree node with metric no bigger than $M$, and its child nodes. So the number of visited nodes by the sphere decoder must be no smaller than that of the TSA.
\end{proof}

According to Theorem \ref{thm:TSA}, the TSA will also visit a polynomial number of nodes on average, as $N\rightarrow \infty$.

\section{Simulation Results}
\label{sec:simulation}
 In this section, we simulate the performance and complexity of the exact ML algorithm for SIMO systems with $N$ receive antennas, under QPSK and nonconstant-modulus 16-QAM. Channel matrix entries are generated as i.i.d complex Gaussian random variables. We investigate the performance of the ML algorithm for $N$= $10$, $50$, $100$, and $500$ receive antennas. We compare the performance of the joint ML non-coherent data detection algorithm with sub-optimal iterative and non-iterative channel estimation and data detection schemes. We use least square (LS) and minimum mean square error (MMSE) channel estimation for the iterative and non-iterative detection schemes (the reader may refer to \cite{LS&MMSE} for the LS and MMSE channel estimation).

 In each channel coherent block, we embed one symbol which is known by the receiver to resolve channel phase ambiguity at layer $T$ of the data sequence. In the non-iterative channel estimation scheme, the receiver estimates the channel vector using this training symbol. Then, the receiver uses this estimated channel vector to detect the remaining $T-1$ transmitted symbols. The iterative suboptimal scheme exploits the detected data vector from the pervious iteration to obtain a new channel estimation, which, in turn, is used for data detection in the current iteration.  The iterative joint channel estimation and data detection scheme runs 100 iterations for each channel coherence block.

 In Figures \ref{LS8}, \ref{LS20}, \ref{MMSE8}, and \ref{MMSE20}, under the QPSK modulation, the symbol error rate (SER) of the ML algorithm is evaluated as a function of SNR for $T=8$ and $20$ respectively, along with the SER of data detection based on the iterative and non-iterative LS and MMSE channel estimations. It can be seen that the ML algorithm outperforms the LS and MMSE iterative and non-iterative channel estimation schemes. For example, from Figures \ref{LS8} and \ref{MMSE8}, we see more than 2 dB improvement over the iterative channel estimation and data detection, and 3 dB improvement over the non-iterative channel estimation and data detection for $N$=$100$, at $10^{-2}$ SER. In Figures \ref{LS20} and \ref{MMSE20}, the ML detector provides a performance improvement of 2 dB over the iterative scheme and 4.5 dB improvement over the non-iterative scheme, at $10^{-2}$ SER.

 We further evaluate the complexities of both sphere decoder and the TSA for QPSK constellation by the average number of visited nodes in each coherence block. In Figure \ref{proposed radius}, we obtain the average number of visited nodes for $T$=$20$ at different SNR values. We use our proposed search radius $r^2=\frac{T}{3}$ for the sphere decoder. It can be seen that when $N$ increases, the number of visited nodes significantly decreases. In fact, the average number of visited nodes for $N$=$500$ is steady at $76$, namely the cardinality of the QPSK constellation multiplied by ($T-1$) layers. This is consistent with our theoretical prediction in Theorem \ref{thm:ltt1}. In addition, the TSA further reduces the complexity, compared with the sphere decoder ML algorithm. At SNR $=-4$ dB, our algorithms on average visit only around several hundred nodes for $N=50$, and only $76$ nodes for $N=500$. In comparison, the exhaustive search method will need to examine $4^{19} \approx 2.75\times 10^{11}$ hypotheses for each coherence block. Our algorithms achieve complexity reduction in many orders of magnitude across a wide range of $N$.

 Figure \ref{NONperfomance} describes the performance of ML channel estimation and data detection algorithm for the nonconstant-modulus 16-QAM constellation. We choose the the coherent time $T=12$, and $N=50, 100$ and $500$. We can see that our novel joint ML algorithms provides nearly $5$ dB gain over iterative joint MMSE channel estimation and data detection algorithms. Under 16-QAM, Figure \ref{NONconstantComplexity} presents the average number of visited nodes, under different SNR values, for sphere decoders with $r^2=\frac{2}{45}$ and for the TSA. The average is taken over $10^3$ channel coherence blocks. Both algorithms achieve surprisingly low average computational complexity. Note that in order to do exhaustive search, one would need to examine $16^{11}$=$1.76 \times 10^{13}$ hypotheses in each coherence block. For SNR above $-4$ dB, on average the TSA visits only $176$ nodes, a $10^{11}$-fold reduction in complexity compared with exhaustive search.

  We further extend our SIMO joint ML channel estimation and data detection algorithm to uplink data detection in massive MIMO systems with $M$ users. These $M$ users employ orthogonal training sequences with length $M$.  First, we estimate the channel using $M$ orthogonal training sequences. Then, based on MMSE channel estimation from training sequences, we use MMSE data detection to decode the transmitted symbols to $\hat{\textbf{S}}^*$, where $\hat{\textbf{S}}^*$ is an matrix of dimension $M \times T$ containing $M$ users'data. Next, we use the detected signal $\hat{\textbf{S}}^*$ to perform MMSE channel estimation again. Now for each user $j$,  after subtracting the interference from the other $(M-1)$ users using their estimated channels and detected data, we perform joint ML channel estimation and data detection (\ref{eq:mixed}) for user $j$ separately. Namely, for user $j$,  the equivalent optimization problem is given as follows:
 $$ \min_{\h_j, \s^*_{j} \in \Omega^T} \| \overline{X}_j-\hat{\h}_{j} \hat{\s}_j^*\|^2, $$
  where $\overline{X}_j={X-\sum^M_{i\neq j} \hat{\h}_{i} \hat{\s}_i^*}$, $1 \leq i,j \leq M$, and $\hat{\h}_{i}$ and $\hat{\s}_i^*$ are estimated channel and detected data for user $i$ respectively. After we have detected $M$ users' data using (\ref{eq:mixed}), we will use the newly detected data to renew MMSE channel estimation for this MIMO system.  We perform MMSE MIMO channel estimation and SIMO joint channel estimation and data detection (\ref{eq:mixed})  iteratively for $10$ times.

Figure \ref{MIMO} shows the performance of this proposed data detection scheme for a massive MIMO system with 4 users, and different numbers of receive antennas at the BS. We employ QPSK modulation, and assume a channel coherence time $T$=$20$. We compare our scheme with iterative MMSE channel estimation and data detection scheme, and non-iterative MMSE channel estimation and data detection. For non-iterative channel estimation and data detection, we will perform one-time MMSE data detection based on the MMSE channel estimation from training sequences.  In iterative MMSE channel estimation and data detection, after we get the detected data from MMSE data detection, we re-estimate the MIMO channel using both training sequences and detected data. This progress is iterated for 10 times. From Figure \ref{MIMO},  we observe that our algorithm employing the SIMO joint channel estimation and data detection algorithm achieves better performance than iterative MMSE channel estimation and data detection. For instance, for $N$=$50$ and SER=$10^{-2}$, our SIMO joint channel estimation and data detection algorithm has roughly $2$ dB gain over non-iterative MMSE channel estimation and data detection, and $1$ dB gain over iterative MMSE channel estimation and data detection scheme. For $N$=$100$, our SIMO joint channel estimation and data detection algorithm has $2$ dB gain over non-iterative MMSE channel estimation and data detection, and $1.5$ dB gain over iterative channel estimation and data detection scheme at the same SER.

\section{Conclusions and Future Work}
\label{sec:conclusion}
To the best of our knowledge, this paper shows, for the first time, the performance of joint ML channel estimation and data detection algorithm of massive SIMO wireless systems, for both constant-modulus and nonconstant-modulus constellations. We have shown that, as the number of receive antennas grows large, the expected complexity of our proposed algorithm is polynomial in the channel coherence time, and the number of receive antennas. Simulation results show that the ML algorithm has better performance than suboptimal noncoherent data detection schemes. In addition, our simulation results verify our theoretical predictions.

It is very interesting to further explore designing efficient joint ML channel estimation and data detection for general massive MIMO systems with multiple users or transmit antennas. Such algorithms will be very useful in reducing pilot contaminations in general massive MIMO systems.

\begin{figure}[!htb]
\centering
\includegraphics[width =3.5 in]{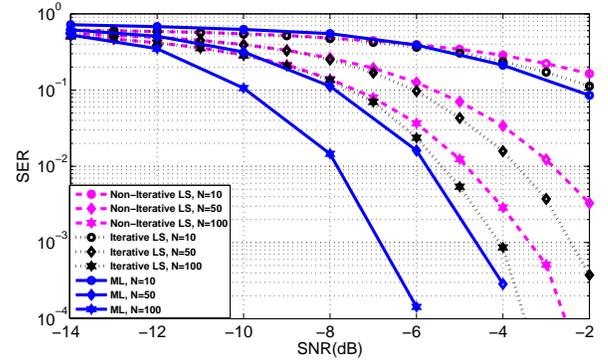}
\caption{SER vs SNR for joint ML channel estimation and data detection, iterative and non-iterative LS channel estimation for $T=8$ and QPSK.}
\label{LS8}
\end{figure}

\begin{figure}[!htb]
\centering
\includegraphics[width =3.5 in]{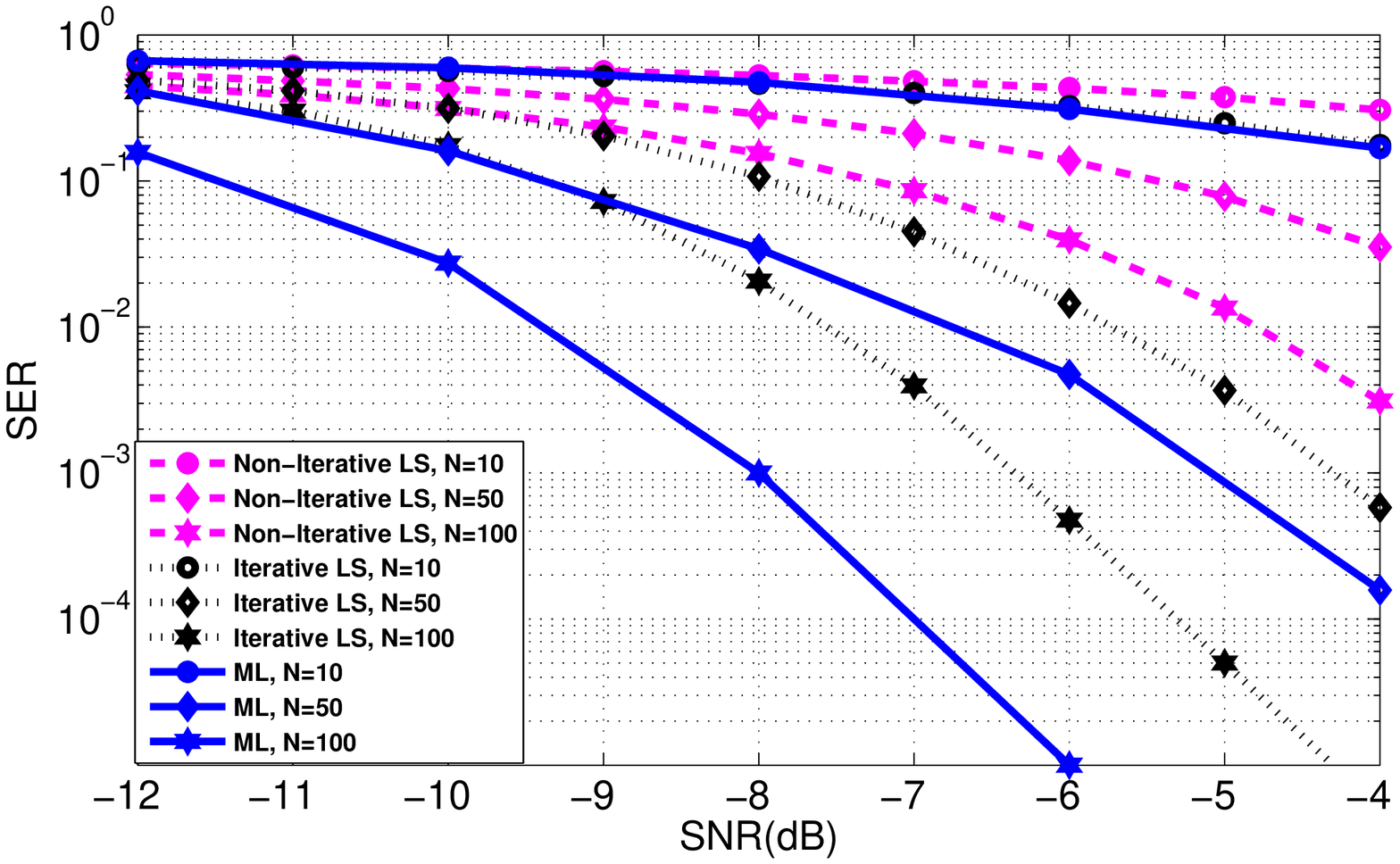}
\caption{SER vs SNR for joint ML channel estimation and data detection, iterative and non-iterative LS channel estimation with $T=20$ and QPSK modulation.}
\label{LS20}
\end{figure}

\begin{figure}[!htb]
\centering
\includegraphics[width =3.5 in]{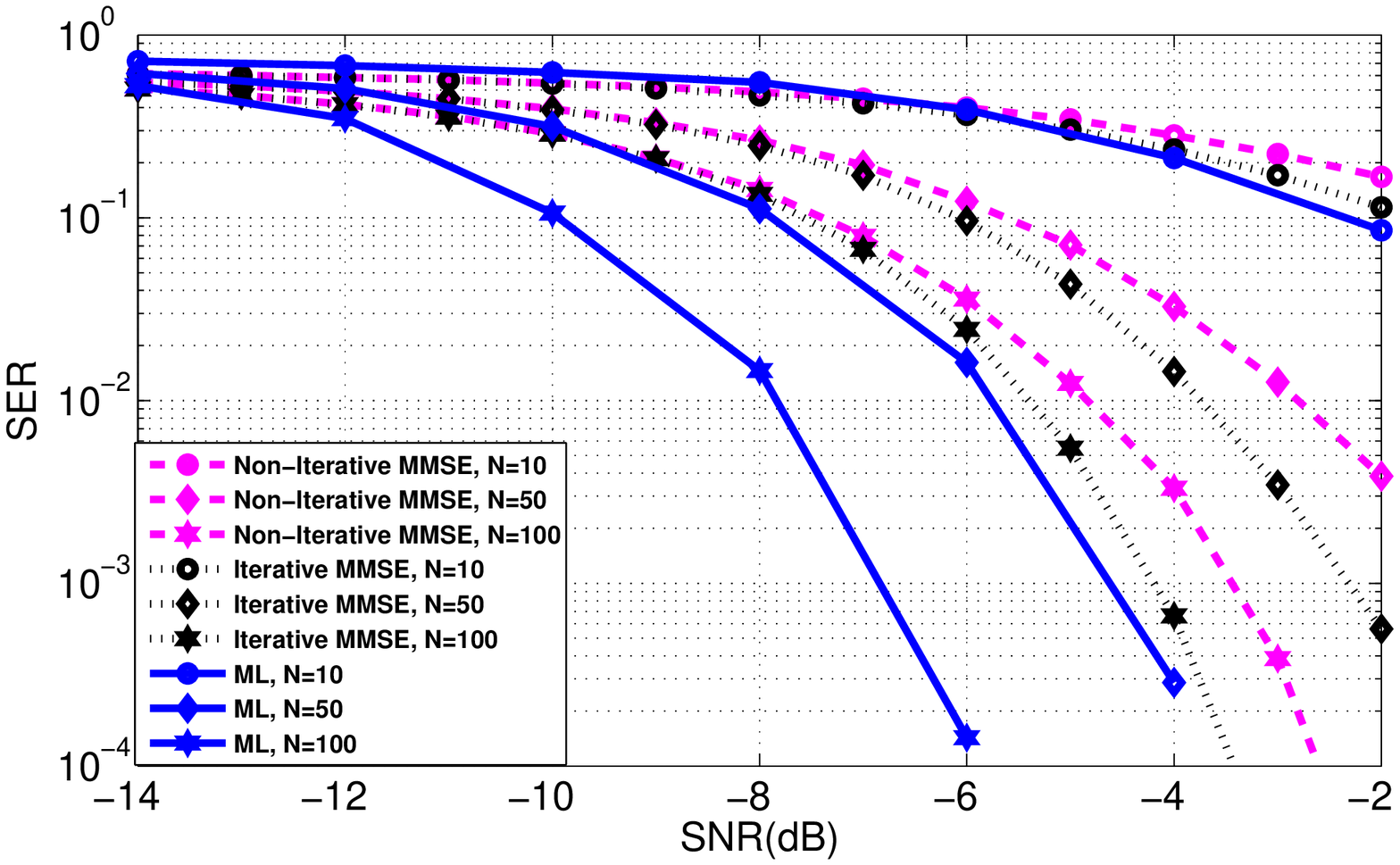}
\caption{SER vs SNR for joint ML channel estimation and data detection, iterative and non-iterative MMSE channel estimation with $T=8$ and QPSK modulation.} \label{MMSE8}
\end{figure}

\begin{figure}[!htb]
\centering
\includegraphics[width =3.5 in]{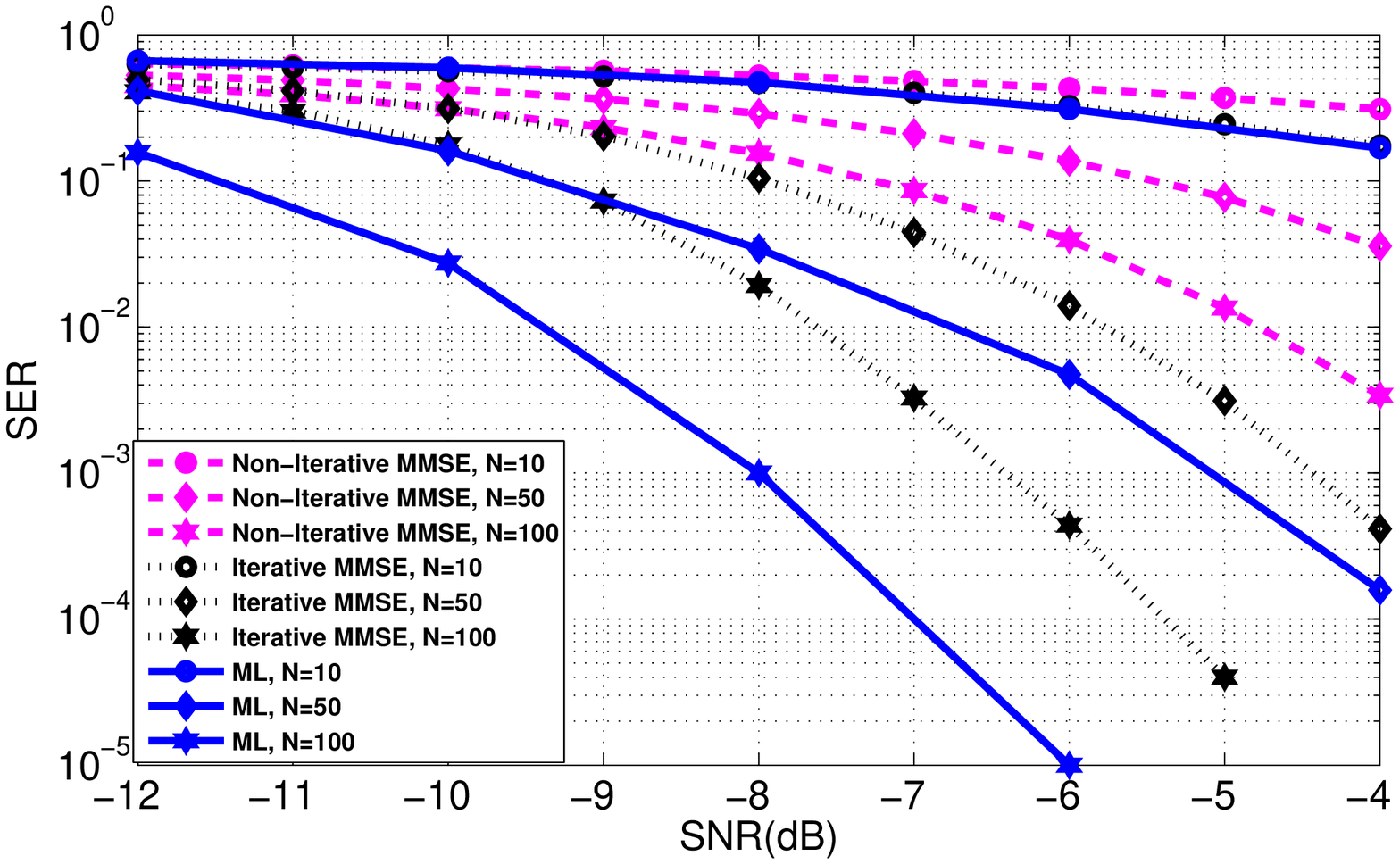}
\caption{SER vs SNR for joint ML channel estimation and data detection, iterative and non-iterative MMSE channel estimation with $T=20$ and QPSK modulation.} \label{MMSE20}
\end{figure}

\begin{figure}[!htb]
  \centering
  \includegraphics[width =3.5 in]{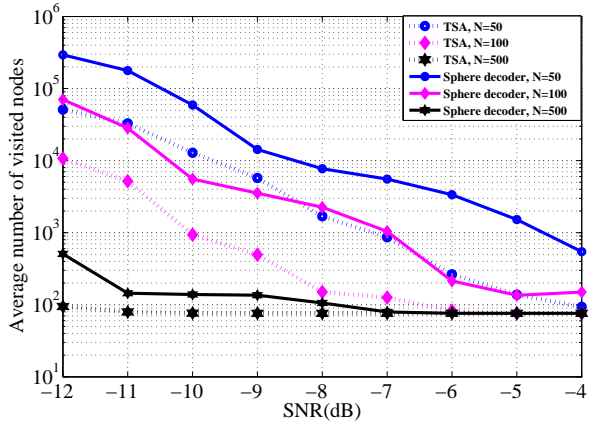}\\
  \caption{Average number of visited points for $T=20$ and QPSK modulation. Exhaustive search will instead need to examine $2.75\times 10^{11}$ hypotheses. }\label{proposed radius}
\end{figure}

\begin{figure}[!htb]
\centering
\includegraphics[width =3.5 in]{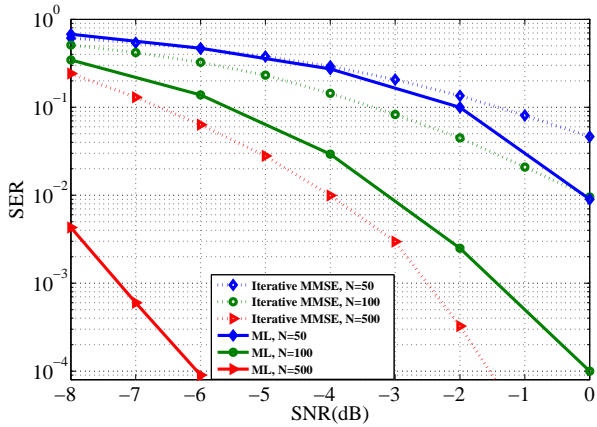}
\caption{SER vs SNR, for joint ML channel estimation and data detection and iterative MMSE channel estimation with $T=12$ and 16-QAM.} \label{NONperfomance}
\end{figure}

\begin{figure}[!htb]
\centering
\includegraphics[width =3.5 in]{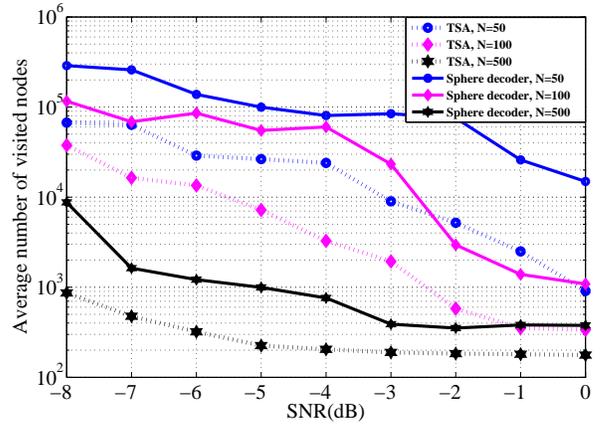}
\caption{Average number of visited points, $T$=12 with 16-QAM. Exhaustive search will instead need to examine $1.76\times 10^{13}$ hypotheses. } \label{NONconstantComplexity}
\end{figure}

\begin{figure}[!htb]
\centering
\includegraphics[width =3.5 in]{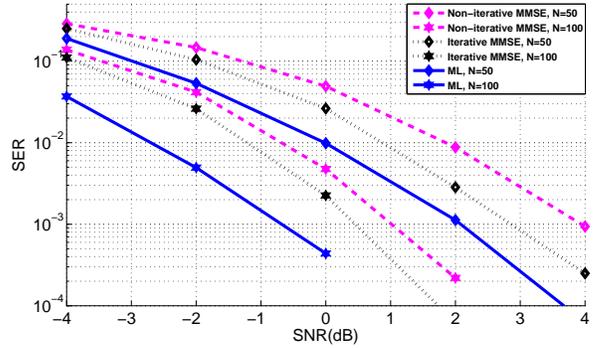}
\caption{SER vs SNR for joint ML channel estimation and data detection, iterative and non-iterative MMSE of MIMO wireless system, $T=14$ and $M=4$}
\label{MIMO}
\end{figure}
\vspace{-3ex}

\appendices
\section{Proof of Lemma \ref{thm:ltt2}}

\begin{proof}
For any $ \widetilde{\s}^* \neq \s^*$, let $i$ be the closest integer to $T$ such that $\s^*_i \neq \widetilde{\s}^*_{i}$, where $1 \leq i \leq T-1$.  Then we can find the metric of $\widetilde{\s}^*_{i:T}$ based on (\ref{eq:Mertici1})
\begin{align}
M_{\widetilde{\s}^*_{i:T}} &= |\sum^{T}_{k=i} L_{i,k} \widetilde{\s}_k|^2+M_{\widetilde{\s}^*_{i+1:T}}\notag\\
  &=|\sum_{k=i+1}^{T} L_{i,k} \s_k+L_{i,i} \widetilde{\s}_i|^2,\notag
\end{align}
where $\widetilde{\s}^*_{i+1:T}=\s^*_{i+1:T}$, and $M_{\widetilde{\s}^*_{i+1:T}}=M_{{\s}^*_{i+1:T}}=0$ as proved in Theorem \ref{thm:ltt1}. Now we can write (\ref{eq:Metric}) as
\begin{align}
M_{\widetilde{\s}^*_{i:T}} &= |\sum_{k=i}^{T} L_{i,k} \s_k- L_{i,i} {\s}_i+ L_{i,i} \widetilde{\s}_i|^2\notag\\
  &=|- L_{i,i} {\s}_i+ L_{i,i} \widetilde{\s}_i|^2\notag\\
  &=|L_{i,i} ({\widetilde{\s}_i-\s}_i)|^2,\notag\\ \nonumber
\end{align}
where we have used the fact that $\sum_{k=i}^{T} L_{i,k} \s_k=0$, as shown in the proof of Theorem \ref{thm:ltt1}. Since ${\widetilde{\s}_i-\s}_i \neq 0$ by assumption, and $L_{i,i}\neq 0$ for $i\neq T$ according to Lemma \ref{thm:ltt3}, $M_{\widetilde{\s}^*_{i:T}}$ will not be zero either.

When $ \widetilde{\s}^* \neq \s^*$,  $M_{\widetilde{\s}^*}$ is thus lower bounded by $|L_{i,i} ({\widetilde{\s}_i-\s}_i)|^2$, $i<T$. The smallest possible value for $|L_{i,i} ({\widetilde{\s}_i-\s}_i)|^2$ is given by $i=T-1$ (see Lemma \ref{thm:ltt3}) and $|({\widetilde{\s}_i-\s}_i)|^2=\min_{s_1,s_2 \in \Omega, s_1\neq s_2} |s_1-s_2|^2$.

\end{proof}

\section{Lemma \ref{thm:ltt3} and its proof}
\begin{lemma}
$L_{i,i}\geq \sqrt{T/2}$ for any $1\leq i \leq T-1$, and $L_{T,T}$ is equal to zero.
\label{thm:ltt3}
\end{lemma}

\begin{proof}
\begin{align}\label{Livalue}
		L_{i,i}&=\sqrt {(T-1)-\sum_{j=1}^{i-1} \frac{T}{(T-(j-1))(T-j)}}\notag\\
			   &=\sqrt {(T-1)+\sum_{j=1}^{i-1} \left( \frac{T}{(T-(j-1))}-\frac{T}{(T-j)}\right)} \notag\\
               &=\sqrt{T-\frac{T}{T-(i-1)}}.\notag
\end{align}
 When $i=T$, (\ref{Livalue}) will be
\begin{align}
		L_{T,t}&=\sqrt{T-\frac{T}{T-(T-1)}}\notag\\
			   &=0.\nonumber
\vspace{-6ex}\end{align}
We can also see that $L_{i,i}\geq \sqrt{T/2}$ for any $i<T$, taking equality when $i=T-1$.
\vspace{-3ex}\end{proof}

\section{Derivation of $var[(X^*X)_{i,j}/N]$ in (\ref{eq:VAR})}
\label{appendix:variance}
\begin{proof}
\begin{align}
&var[(X^*X)_{i,j}] \notag\\
&=var[\sum^N_{k=1}{B_k  }]
= \sum^N_{k=1}var(B_{k})\notag\\
&=\sum^N_{k=1}({E}[B_k B_k^*]-{E}[B_k]{E}[B_k^*])\nonumber
\end{align}
where $B_{k}=(\s_i^*\h_k +\w_{k,i})^*(\s_j^*\h_k +\w_{k,j})$. By expansion, we have
\begin{flalign*}
  E[B_k B_k^*] & =\underbrace{\s_i\s^*_j\s^*_i\s_j}_{=1} \h^*_k \h_k\h^*_k \h_k+\s_i\s^*_j\h^*_k \h_k \w^*_{k,j} \w_{k,i}\nonumber\\
 &+ \underbrace{\s_i\s^*_j \s^*_i \h^*_k \h_k\h_k \w^*_{k,j}}_{=0}+\underbrace{\s_i\s^*_j \s_j \h^*_k \h_k\w_{k,i} \h^*_k}_{=0} \nonumber\\
 &+ \s_j\s^*_i\w^*_{k,i} \w_{k,j}\h^*_k \h_k+\w^*_{k,i} \w_{k,j}\w^*_{k,j} \w_{k,i}\nonumber\\
 &+\> \underbrace{\s^*_i\w^*_{k,i} \w_{k,j}\h_k \w^*_{k,j}}_{=0}+\underbrace{\s_j\w^*_{k,i} \w_{k,j}\w_{k,i} \h^*_k}_{=0}\notag\\
 &+\> \underbrace{\s_i\s_j\s^*_i \h^*_k \w_{k,j}\h^*_k \h_k}_{=0}+\underbrace{\s_i \h^*_k \w_{k,j}\w^*_{k,j} \w_{k,i}}_{=0}\notag\\
 &+\> \underbrace{\s_i \s^*_i}_{=1}\h^*_k \w_{k,j}\h_k \w^*_{k,j}+\s_i \s_j\h^*_k \w_{k,j}\w_{k,i} \h^*_k\notag\\
 &+\> \underbrace{\s^*_j \s_j\s^*_i\w^*_{k,i} \h_k\h^*_k \h_k}_{=0}+\underbrace{\s^*_j \w^*_{k,i} \h_k\w^*_{k,j} \w_{k,i}}_{=0}\notag\\
 &+\> \s^*_j \s^*_i\w^*_{k,i} \h_k\h_k \w^*_{k,j}+\underbrace{\s^*_j \s_j}_{=1} \w^*_{k,i} \h_k\w_{k,i} \h^*_k\notag\\
 \end{flalign*}
 Since we already assume that the entries of $\h$ are rotationally-invariant complex Gaussian with unit variance, then we can write $\h_k$ as $a+b\sqrt{-1}$, where $a$ and $b$ are independent, and both follow Gaussian distribution $\mathcal{N}(0,\frac{1}{2})$. Thus ${E}[\h_{k}^2]={E}[(\h_{k}^*)^2]=0$.
Furthermore,
\begin{flalign}
  {{E}[|\h_k|^4]}&= {E}[(a^2+b^2)^2]={E}[a^4+b^4+2 a^2 b^2]\notag\\
&=3\sigma^4_{a}+3\sigma^4_{b}+2\sigma^2_{a}\sigma^2_{b}\notag\\
&=2\times 3\times(\frac{1}{2})^2+\frac{2}{4}=2,
\end{flalign}
where $\sigma^2_{a}=\frac{1}{2}$ and $\sigma^2_{b}=\frac{1}{2}$ are respectively the variance of $a$ and $b$. In the same way, we can find ${{E}[|\w|^4]}=2\sigma^4_w$.

Thus, when $i\neq j$,
\begin{flalign}
  {E}[B_{k} B_{k}^*] &= {E}[|\h_k|^4]+{E}[|\w_{k,i}|^2]{E}[|\w_{k,j}|^2]\notag\\
  &+{E}[|\h_k|^2]{E}[|\w_{k,i}|^2]+{E}[|\h_k|^2]{E}[|\w_{k,j}|^2]\notag\\
  &=2+\sigma^4_{w}+2\sigma^2_{w}.
\end{flalign}

When $i=j$,
\begin{flalign}
  {E}[B_k B_k^*] &= \underbrace{{E}[|\h_k|^4]}_{=2}+\underbrace{{E}[|\w_{k,i}|^4]}_{=2\sigma^4_{w}}\notag\\
  &+\underbrace{{E}[|\h_k|^2]{E}[|\w_{k,i}|^2]}_{=\sigma^2_{w}}+\underbrace{{E}[|\h_k|^2]{E}[|\w_{k,i}|^2]}_{=\sigma^2_{w}}\notag\\
  &+\underbrace{{E}[|\h_k|^2]{E}[|\w_{k,i}|^2]}_{=\sigma^2_{w}}+\s^2_{i}\underbrace{{E}[(\h_k^*)^2]{E}[(\w_{k,i})^2]}_{=0}\notag\\
  &+(\s^2_{i})^*\underbrace{{E}[(\h_k)^2]{E}[(\w^*_{k,i})^2]}_{=0}+\underbrace{{E}[|\h_k|^2]{E}[|\w_{k,i}|^2]}_{=\sigma^2_{w}}\notag\\
  &=2+2\sigma^4_{w}+4\sigma^2_{w}.\notag\\
\end{flalign}

Moreover, after some algebra,
$$E[B_k]E[B_k^*]=\begin{cases}
    1+2\sigma^2_{w}+\sigma^4_{w},  & \text{if } i= j\\
    \s_{i}\s^*_{j}\s_{j}\s^*_{i}=1,  & \text{otherwise}.
\end{cases}$$

Finally,
\begin{flalign}
var(B_k)&={E}[B_k B_k^*]-{E}[B_k]{E}[B_k^*]\notag\\
&=\begin{cases}
    1+2\sigma^2_{w}+\sigma^4_{w},  & \text{if } i= j\\
    1+2\sigma^2_{w}+\sigma^4_{w},  & \text{otherwise}
\end{cases}
\end{flalign}
This leads to
\begin{equation}
var(\frac{(X^*X)_{i,j}}{N})=(1+2\sigma^2_{w}+\sigma^4_{w})/N.
\end{equation}

\end{proof}

\section{Proof of Lemma \ref{thm:radiusthm_nonconstant}}
\label{sec:appendxiradius16QAM}
\begin{proof}
Let us recall that $t=\sum^{T}_{i=1} \|\s_{i}^*\|^2$.

\begin{flalign}
  L_{i,i}&=\sqrt{t-||\s_{i}^*||^2-\sum^{i-1}_{j=1}\frac{||\s^*_{j}||^2||\s^*_{i}||^2t}{(t-||\s^*_{1:j-1}||^2)(t-||\s^*_{1:j}||^2)}}\notag\\
 &=\sqrt{t-||\s^*_{i}||^2+\sum^{i-1}_{j=1}[\frac{||\s^*_{1:j-1}||^2||\s^*_{i}||^2}{t-||\s^*_{1:j-1}||^2}-\frac{||\s^*_{1:j}||^2||\s^*_{i}||^2}{t-||\s^*_{1:j}||^2}]}\notag\\
 &=\sqrt{t-||\s^*_{i}||^2-\frac{||\s^*_{1:i-1}||^2||\s^*_{i}||^2}{t-||\s^*_{1:i-1}||^2}}\notag\\
 &=\sqrt{t(1-\frac{||\s^*_{i}||^2}{||\s^*_{i:T}||^2})} \label{Lem-non}
\end{flalign}
We can see that for any $i\neq T$, $\frac{||\s^*_{i}||^2}{||\s^*_{i:T}||^2}\neq 1$ and thus $L_{i,i} \neq 0$. However, when $i=T$,
$$L_{T,T}=\sqrt{t(1-\frac{||\s^*_{T}||^2}{||\s^*_{T:T}||^2})}=0.$$

 For any $ \widetilde{\s}^*$ such that $\widetilde{\s}^*\neq \s^*$, let $i$ be the largest integer such that $ \s^*_{i} \neq \widetilde{\s}^*_{i}$. Then for any $j\leq i$,

$$\bar{M}_{\widetilde{\s}_{j:T}^*} \geq  \frac{L_{i,i}^2}{\|\s^*_{j:T}\|^2+18(j-1)} \|\widetilde{\s}^*_{i}-\s^*_{i} \|^2.$$

We would like to give a lower bound on the right side of the equation above. We first lower bound $L_{i,i}^2=t(1-\frac{||\s^*_{i}||^2}{||\s^*_{i:T}||^2})$. The smallest possible value for $t$ is $t=2T$ (achieved when every symbol is in the form of $\pm 1 \pm j$ ), and the largest possible value for $\frac{||\s_{i}||^2}{||\s_{i:T}||^2}$ is $i=T-1$, $\| \s_{T-1}\|^2=18$, and $\|\s_{T}\|^2=2$. Thus $L_{i,i}^2$ is lower bounded by $2T(1-\frac{18}{18+2})=T/5$.  Furthermore, the smallest possible value for  $\|\widetilde{\s}^*_{i}-\s^*_{i} \|^2=4$, and the largest possible value for $\|\s^*_{j:T}\|^2+18(j-1)$ is $18 T$. This in turn gives $\bar{M}_{\widetilde{\s}_{j:T}^*} $ a lower bound of $4 \times (T/5) /(18T)=2/45.$
\end{proof}

\bibliographystyle{IEEEtran}
\bibliography{refs}

\end{document}